\documentclass{amsart}
\usepackage{mathabx,mathtools,amsrefs,mathrsfs,math}
\newcommand\K{{\mathbb K}}
\numberwithin{equation}{section}

\begin{document}
\title{Amenability of groups is characterized by Myhill's Theorem}
\author{Laurent Bartholdi\\ \lowercase{\textit{with an appendix by}} Dawid Kielak}
\date{June 1, 2016}
\address{L.B.: D\'epartement de Math\'ematiques et Applications, \'Ecole Normale Sup\'erieure, Paris \textit{and} Mathematisches Institut, Georg-August Universit\"at zu G\"ottingen}
\email{laurent.bartholdi@gmail.com}
\address{D.K.: Fakult\"at für Mathematik, Universit\"at Bielefeld}
\email{dkielak@math.uni-bielefeld.de}

\thanks{This work is supported by the ``@raction'' grant ANR-14-ACHN-0018-01}
\begin{abstract}
  We prove a converse to Myhill's ``Garden-of-Eden'' theorem and
  obtain in this manner a characterization of amenability in terms of
  cellular automata: \emph{A group $G$ is amenable if and only if
    every cellular automaton with carrier $G$ that has gardens of Eden
    also has mutually erasable patterns.}

  This answers a question by Schupp, and solves a conjecture by
  Ceccherini-Silberstein, Mach\`\i\ and Scarabotti.

  An appendix by Dawid Kielak proves that group rings without zero
  divisors are Ore domains precisely when the group is amenable,
  answering a conjecture attributed to Guba.
\end{abstract}
\maketitle

\section{Introduction}
Cellular automata were introduced in the late 1940's by von Neumann as
models of computation and of biological
organisms~\cite{vneumann:generalautomata}. We follow an algebraic
treatment, as in~\cite{ceccherini-coornaert:cag}: let $G$ be a
group. A \emph{cellular automaton} carried by $G$ is a $G$-equivariant
continuous map $\Theta\colon A^G\to A^G$ for some finite set
$A$. Elements of $A^G$ are called \emph{configurations}, and the
action of $G$ on $A^G$ is given by
\[g\cdot \phi=\phi(-\cdot g)\text{ for all }\phi\in A^G,g\in G.\]

One should think of $A$ as the stateset (e.g.\ ``\textsf{asleep}'' or
``\textsf{awake}'') of a microscopic animal; then $A^G$ is the
stateset of a homogeneous swarm of animals indexed by $G$, and
$\Theta$ is an evolution rule for the swarm: it is identical for each
animal by $G$-equivariance, and is only based on local interaction by
continuity of $\Theta$. For example, fixing $f,\ell,r\in G$ the
``front'', ``left'' and ``right'' neighbours, define $\Theta$ by
``sleep if the guy in front of you sleeps, unless both your neighbours
are awake'', or in formul\ae, set for all $\phi\in A^G,g\in G$
\[\Theta(\phi)(g)=\begin{cases}\textsf{asleep} & \text{ if }\phi(f g)=\textsf{asleep}\text{ and }\{\phi(\ell g),\phi(r g)\}\ni\textsf{asleep},\\ \phi(g) & \text{ else.}\end{cases}\]
Generally speaking, the \emph{memory set} of a cellular automaton is
the minimal $S\subseteq G$ such that $\Theta(\phi)(g)$ depends only on
the restriction of $\phi$ to $Sg$, and is finite.

Two properties of cellular automata received particular attention. Let
us call \emph{pattern} the restriction of a configuration to a finite
subset $Y\subseteq G$. On the one hand, there can exist patterns that
never appear in the image of $\Theta$. These are called \emph{Gardens
  of Eden} (GOE), the biblical metaphor expressing the notion of
paradise the universe may start in but never return to.

On the other hand, $\Theta$ can be non-injective in a strong sense:
there can exist patterns $\phi'_1\neq\phi'_2\in A^Y$ such that,
however one extends $\phi'_1$ to a configuration $\phi_1$, if one
extends $\phi'_2$ similarly (i.e.\ in such a way that $\phi_1$ and
$\phi_2$ have the same restriction to $G\setminus Y$) then
$\Theta(\phi_1)=\Theta(\phi_2)$. These patterns $\phi'_1,\phi'_2$ are
called \emph{Mutually Erasable Patterns} (MEP).  Equivalently there
are two configurations $\phi_1,\phi_2$ which differ on a non-empty
finite set and satisfy $\Theta(\phi_1)=\Theta(\phi_2)$. The absence of
MEP is sometimes called
\emph{pre-injectivity}~\cite{gromov:endomorphisms}*{\S8.G}.

Amenability of groups was also introduced by von Neumann, in the late
1920's in~\cite{vneumann:masses}; there exist numerous formulations
(see e.g.~\cite{wagon:banachtarski}), but we content ourselves with
the following criterion due to F\o lner (see~\cite{folner:banach})
which we treat as a definition: \emph{a discrete group $G$ is amenable
  if for every $\epsilon>0$ and every finite $S\subset G$ there exists
  a finite $F\subseteq G$ with $\#(S F)<(1+\epsilon)\#F$.} In words,
there exist finite subsets of $G$ that are arbitrarily close to
invariant under translation.

Cellular automata were initially considered on $G=\Z^n$. Celebrated
theorems by Moore and Myhill~\cites{moore:ca,myhill:ca} prove that, in
this context, a cellular automaton admits GOE if and only if it admits
MEP; necessity is due to Myhill, and sufficiency to Moore. This result
was generalized by Mach\`\i\ and Mignosi~\cite{machi-m:ca} to groups
of subexponential growth, and by Ceccherini-Silberstein, Mach\`\i\ and
Scarabotti~\cite{ceccherini-m-s:ca} to amenable groups.

\noindent Our main result is a converse to Myhill's theorem:
\begin{thm}\label{thm:main}
  Let $G$ be a non-amenable group. Then there exists a cellular
  automaton carried by $G$ that admits Gardens of Eden but no mutually
  erasable patterns.
\end{thm}

There is a natural measure, the \emph{Bernoulli measure}, on the
configuration space $A^G$: for every pattern $\phi\in A^Y$ it assigns
measure $1/\#A^{\#Y}$ to the clopen set
$\{\psi\in A^G:\,\psi|Y=\phi\}$. Note that the $G$-action on $A^G$
preserves this measure. Hedlund proved
in~\cite{hedlund:endomorphisms}*{Theorem~5.4}, for $G=\Z$, that a
cellular automaton preserves Bernoulli measure if and only if it has
no GOE. This result was generalized by Meyerovitch to amenable
groups~\cite{meyerovitch:finiteentropy}*{Proposition~5.1}.

Combining these with Theorem~\ref{thm:main} and with the
aforementioned results by Ceccherini-Silberstein \emph{et al.} and the
main result of~\cite{bartholdi:moore}, we deduce:
\begin{cor}\label{cor:main}
  Let $G$ be a group; then the following are equivalent:
  \begin{enumerate}
  \item the group $G$ is amenable;\label{cor:1}
  \item all cellular automata on $G$ that admit MEP also admit
    GOE;\label{cor:2}
  \item all cellular automata on $G$ that admit GOE also admit
    MEP;\label{cor:3}
  \item all cellular automata on $G$ that do not preserve Bernoulli
    measure admit GOE.\qed\label{cor:4}
  \end{enumerate}
\end{cor}

\subsection{Origins}
Schupp had already asked in~\cite{schupp:arrays}*{Question~1} in which
precise class of groups the theorems by Moore and Myhill hold.
Ceccherini-Silberstein \emph{et al.} conjecture
in~\cite{ceccherini-m-s:ca}*{Conjecture 6.2} that
Corollary~\ref{cor:main}(1--3) are equivalent.

The implication~(3$\Rightarrow$1) is the content of
Theorem~\ref{thm:main}. In case $G$ contains a non-abelian free
subgroup, it was already shown by Muller in his University of Illinois
1976 class notes, see~\cite{machi-m:ca}*{page 55}; let us review the
construction, in the special case
$G=\langle x,y,z|x^2,y^2,z^2\rangle$. Fix a finite field $\K$, and set
$A\coloneqq\K^2$. View $A^G$ as $\K^G\times\K^G$, on which $2\times2$ matrices
with co\"efficients in the group ring $\K G$ act from the left. Define
$\Theta\colon A^G\righttoleftarrow$ by
\[\Theta(\phi)=\begin{pmatrix}x & y+z\\ 0 & 0\end{pmatrix}\phi.\]
It obviously has gardens of Eden --- any pattern with non-trivial
second co\"ordinate --- and to show that it has no mutually erasable
patterns it suffices, since $\Theta$ is linear, to show that $\Theta$
is injective on finitely-supported configurations; this is easily
achieved by considering, in the support of a configuration $\phi$, a
position $g\in G$ such that $x g$ and $y g$ don't belong to the support
of $\phi$.

\subsection*{Acknowledgments}
I am very grateful to Tullio Ceccherini-Silberstein and to Alexei
Kanel-Belov for entertaining conversations and encouragement, and to
Dawid Kielak for having contributed an appendix to the text.

\section{Proof of Theorem~\ref{thm:main}}
\noindent We begin with a combinatorial
\begin{lem}\label{lem:1}
  Let $n$ be an integer. Then there exists a set $Y$ and a family of
  subsets $X_1,\dots,X_n$ of $Y$ such that, for all
  $I\subseteq\{1,\dots,n\}$ and all $i\in I$, we have
  \[\#\Big(X_i\setminus\bigcup_{j\in I\setminus\{i\}}X_j\Big)\ge\frac{\#Y}{(1+\log n)\#I}.\]
\end{lem}
\begin{proof}
  We denote by $\sym n$ the symmetric group on $n$ letters.  Define
  \[Y\coloneqq\frac{\{1,\dots,n\}\times\sym n}{(i,\sigma)\sim(j,\sigma)\text{ if $i$ and $j$ belong to the same cycle of }\sigma};
  \]
  in other words, $Y$ is the set of cycles of elements of $\sym
  n$. Let $X_i$ be the natural image of $\{i\}\times\sym n$ in the
  quotient $Y$.

  First, there are $(i-1)!$ cycles of length $i$ in $\sym i$, given by
  all cyclic orderings of $\{1,\dots,i\}$; so there are
  $\binom n i(i-1)!$ cycles of length $i$ in $\sym n$, and they can be
  completed in $(n-i)!$ ways to a permutation of $\sym n$; so
  \begin{equation}\label{eq:lem:1}
    \#Y=\sum_{i=1}^n\binom n i(i-1)!(n-i)!=\sum_{i=1}^n\frac{n!}i\le(1+\log n)n!
  \end{equation}
  since $1+1/2+\dots+1/n\le 1+\log n$ for all $n$.

  Next, consider $I\subseteq\{1,\dots,n\}$ and $i\in I$, and set
  $X_{i,I}\coloneqq X_i\setminus\bigcup_{j\in
    I\setminus\{i\}}X_j$. Then
  $X_{i,I}=\big\{(i,\sigma):(i,\sigma)\nsim(j,\sigma)\text{ for all
  }j\in I\setminus\{i\}\big\}$. Summing over all possibilities for the
  length-$(j+1)$ cycle $(i,t_1,\dots,t_j)$ of $\sigma$ intersecting
  $I$ in $\{i\}$, we get
  \begin{equation}\label{eq:lem:2}
    \begin{split}
      \#X_{i,I}&=\sum_{j=0}^{n-\#I}\binom{n-\#I}j j!(n-j-1)!\\
      &=\sum_{k\coloneqq n-j=\#I}^n(n-\#I)!(\#I-1)!\binom{k-1}{k-\#I}\\
      &=(n-\#I)!(\#I-1)!\binom{n}{n-\#I}=\frac{n!}{\#I}.
    \end{split}
  \end{equation}
  Combining~\eqref{eq:lem:1} and~\eqref{eq:lem:2}, we get
  \[\#X_{i,I}=\frac{n!}{\#I}=\frac{(1+\log n)n!}{(1+\log n)\#I}\ge\frac{\#Y}{(1+\log n)\#I}.\qedhere\]
\end{proof}

Let $G$ be a non-amenable group. To prove Theorem~\ref{thm:main}, we
construct a cellular automaton carried by $G$, with GOE but without
MEP. Since $G$ is non-amenable, there exists $\epsilon>0$ and
$S_0\subset G$ finite with $\#(S_0F)\ge(1+\epsilon)\#F$ for all finite
$F\subset G$. We then have $\#(S_0^k F)\ge(1+\epsilon)^k\#F$ for all
$k\in\N$. Let $k$ be large enough so that
$(1+\epsilon)^k>1+k\log\#S_0$, and set $S\coloneqq S_0^k$ and
$n\coloneqq\#S$. This set $S$ will be the memory set of our
automaton. We then have
\begin{equation}\label{eq:folner}
  \begin{split}
    \#(S F)&\ge(1+\epsilon)^k\#F>(1+k\log\#S_0)\#F\\
    &\ge(1+\log n)\#F\text{ for all finite }F\subset G.
  \end{split}
\end{equation}

Apply Lemma~\ref{lem:1} to this $n$, and identify $\{1,\dots,n\}$ with
$S$ to obtain a set $Y$ and subsets $X_s$ for all $s\in S$. We have
\[\#\bigg(X_s\setminus\bigcup_{t\in
    T\setminus\{s\}}X_t\bigg)\ge\frac{\#Y}{(1+\log n)\#T}\text{ for
    all }s\in T\subseteq S.
\]
Furthermore, since $n\ge2$ these inequalities are sharp; so we may
replace $Y$ and $X_s$ respectively by $Y\times\{1,\dots,k\}$ and
$X_s\times\{1,\dots,k\}$ for some $k$ large enough so that
$\#(X_s\setminus\bigcup_{t\in T\setminus\{s\}}X_t)\ge(\#Y+1)/(1+\log
n)\#T$ holds; and then we replace $Y$ by $Y\sqcup\{\cdot\}$. If for
$T\subseteq S$ and $s\in T$ we define
\[X_{s,T}\coloneqq X_s\setminus\bigcup_{t\in T\setminus\{s\}}X_t,\text{ then }
  \#X_{s,T}\ge\frac{\#Y}{(1+\log n)\#T}\text{ for all }s\in T\subseteq S;\]
and furthermore we have obtained $\bigcup_{s\in S}X_s\subsetneqq Y$.

Let $\K$ be a large enough finite field (in a sense to be precised
soon), and set $A\coloneqq \K Y$. For each $s\in S$, choose a linear
map $\alpha_s\colon A\to\K X_s\subset A$, and for $T\ni s$ denote by
$\alpha_{s,T}\colon A\to\K X_{s,T}$ the composition of $\alpha_s$ with
the co\"ordinate projection $\pi_{s,T}\colon A\to\K X_{s,T}$, in such
a manner that, whenever $\{T_s: s\in S\}$ is a family of subsets of
$S$ with $\sum_{s\in S}\#X_{s,T_s}\ge\#Y$, we have
\begin{equation}\label{eq:inj}
  \bigcap_{s\in S}\ker(\alpha_{s,T_s})=0.
\end{equation}

This is always possible if $\K$ is large enough: indeed write each
$\alpha_s$ as a $\#Y\times\#Y$ matrix and each $\alpha_{s,T}$ as a
submatrix. The condition is then that various vertical concatenations
of submatrices have full rank, and the complement of these conditions
is a proper algebraic subvariety of $\K^{Y\times Y\times S}$ defined
over $\Z$, which is not full as soon as $\K$ is large enough.

Define now a cellular automaton with stateset $A$ and carrier $G$ by
\[\Theta(\phi)(g) = \sum_{s\in S}\alpha_s(\phi(s g)).\]

Clearly $\Theta$ admits gardens of Eden: for every $\phi\in A^G$, we
have $\Theta(\phi)(1)\in\K(\bigcup_{s\in S}X_s)\subsetneqq A$.

To show that $\Theta$ admits mutually erasable patterns, it is enough
to show, for $\phi\in A^G$ non-trivial and finitely supported, that
$\Theta(\phi)\neq0$. Let thus $F\neq\emptyset$ denote the support of
$\phi$. Define $\rho\colon S F\to(0,1]$ by
$\rho(g)\coloneqq1/\#\{s\in S:g \in s F\}$. Now
\[\sum_{f\in F}\Big(\sum_{s\in S}\rho(s f)\Big)=\sum_{g\in S F}\sum_{s\in S:g\in s F}\rho(g)=\sum_{g\in S F}1=\#(S F),
\]
so there exists $f\in F$ with
$\sum_{s\in S}\rho(s f)\ge\#(S F)/\#F\ge1+\log n$
by~\eqref{eq:folner}. For every $s\in S$, set
$T_s\coloneqq\{t\in S:s f\in t F\}$, so $\#T_s=1/\rho(s f)$. We obtain
\begin{align*}
  \sum_{s\in S}\#X_{s,T_s}&\ge\sum_{s\in S}\frac{\#Y}{(1+\log n)\#T_s}\text{ by Lemma~\ref{lem:1}}\\
  &=\sum_{s\in S}\frac{\#Y\rho(s f)}{1+\log n}\ge\#Y,
\end{align*}
so by~\eqref{eq:inj} the map
$A\ni a\mapsto(\alpha_{s,T_s}(a))_{s\in S}$ is injective. Set
$\psi\coloneqq\Theta(\phi)$. Since by assumption $\phi(f)\neq0$, we
get $(\pi_{s,T_s}(\psi(s f)))_{s\in S}\neq0$, so $\psi\neq0$ and we
have proven that $\Theta$ admits no mutually erasable patterns. The
proof is complete.

\appendix
\section{A characterization of amenability via Ore domains, by Dawid Kielak}
Let $A$ be an associative ring without zero divisors, and let us write
$A^*=A\setminus\{0\}$. Recall that $A$ is called an \emph{Ore domain}
if it satisfies Ore's condition: \emph{for every $a\in A,s\in A^*$
  there exist $b\in A,t\in A^*$ with $at=bs$.} It then follows that
$A(A^*)^{-1}$, namely the set of expressions of the form $as^{-1}$
with $a\in A,s\in A^*$ up to the obvious equivalence relation
$as^{-1}=at(st)^{-1}$, is a skew field called $A$'s \emph{classical
  field of fractions}.

A folklore conjecture, sometimes attributed to Victor
Guba~\cite{guba:question}, asserts that group rings satisfy the Ore
condition precisely when the group is amenable. We prove it in the
following form:
\begin{thm}
  Let $G$ be a group, and let $\K$ be a field such that $\K G$ has no
  zero divisors. Then $G$ is amenable if and only if $\K G$ is an Ore
  domain.
\end{thm}
\begin{proof}
  $(\Rightarrow)$ is due to Tamari~\cite{tamari:folner}; we repeat it
  for convenience.  Assume that $G$ is amenable, and let
  $a\in\K G,s\in(\K G)^*$ be given. Let $S\subseteq G$ be a finite set
  containing the supports of $a$ and $s$. By F\o lner's criterion,
  there exists $F\subseteq G$ finite such that $\#(SF)<2\#F$. Consider
  $b,t\in\K F$ as variables; then the equation system $as=bt$ is
  linear, has $2\#F$ unknowns, and at most $\#(SF)$ equations, so has
  a non-trivial solution.

  $(\Leftarrow)$ Assume that $G$ is non-amenable. The construction in
  the proof of Theorem~\ref{thm:main} yields a finite extension
  $\mathbb L$ of $\mathbb K$ and an $n \times n$ matrix $M$ over
  $\mathbb L G$ such that multiplication by $M$ is an injective map
  $(\mathbb L G)^n\righttoleftarrow$ and $M$'s last row consists
  entirely of zeros. Forgetting that last row and restricting scalars,
  namely writing $\mathbb L=\mathbb K^d$ qua $\K$-vector space, we
  obtain an exact sequence of free $\K G$-modules
  \begin{equation}\label{main eqn}
    0 \longrightarrow (\K G)^{dn} \longrightarrow (\K G)^{d(n-1)}.
  \end{equation}

  Suppose now that $\K G$ is an Ore domain, with classical field of
  fractions $\mathbb F$. Crucially, $\mathbb F$ is a flat $\K G$
  module, that is the functor $-\otimes_{\K G} \mathbb F$ preserves
  exactness of sequences (see
  e.g.~\cite{mcconnell-robson:nnr}*{Proposition~2.1.16}).  Also,
  $\mathbb F$ is a skew field, and upon tensoring \eqref{main eqn}
  with $\mathbb F$ we obtain an exact sequence
  \[0 \longrightarrow \mathbb F^{dn} \longrightarrow \mathbb F^{d(n-1)}\]
  which is impossible for reasons of dimension.
\end{proof}

\end{document}